\documentclass[11pt]{article}

\usepackage[ruled]{algorithm2e}

\usepackage{amsmath,amsfonts,amssymb,bbm,amsthm}

\usepackage{framed}
\usepackage{amssymb}
\usepackage{fullpage}
\usepackage{amsfonts}
\usepackage{amsmath}
\usepackage{cancel}
\usepackage{color}
\usepackage{fancyhdr}
\usepackage{mathrsfs}
\usepackage{latexsym}
\usepackage{epsfig}
\usepackage{dsfont}

\usepackage{caption}

\SetArgSty{textrm} 
\SetAlFnt{\small}
\SetAlCapFnt{\small}
\SetAlCapNameFnt{\small}
\SetAlCapHSkip{0pt}
\IncMargin{-\parindent}

\DeclareMathOperator*{\argmin}{arg\!\min}
\DeclareMathOperator*{\argmax}{arg\!\max}

\newtheorem{definition}{Definition}
\newtheorem{example}{Example}
\newtheorem{theorem}{Theorem}

\newtheorem{lemma}{Lemma}

\newtheorem{corollary}{Corollary}

\begin{document}

\title{Selfish Knapsack} 

\date{}

\author{
Itai Feigenbaum\thanks{IEOR Department, Columbia University, New York,
NY;
{\tt itai@ieor.columbia.edu}} \and
Matthew P. Johnson\thanks{Department of Mathematics and Computer Science, Lehman College and PhD Program in Computer Science, Graduate Center, CUNY, New York,
NY;
{\tt mpjohnson@gmail.com}
}
}
\maketitle

\begin{abstract}
We consider a selfish variant of the knapsack problem. In our version, the items are owned by agents, and each agent can misrepresent the set of items she owns---either by avoiding reporting some of them (understating), or by reporting additional ones that do not exist (overstating). Each agent's objective is to maximize, within the items chosen for inclusion in the knapsack, the total valuation of her own chosen items. The knapsack problem, in this context, seeks to minimize the worst-case approximation ratio for social welfare at equilibrium. We show that a randomized greedy mechanism has attractive strategic properties: in general, it has a correlated price of anarchy of $2$ (subject to a mild assumption). For overstating-only agents, it becomes strategyproof; we also provide a matching lower bound of $2$ on the (worst-case) approximation ratio attainable by randomized strategyproof mechanisms, and show that no deterministic strategyproof mechanism can provide any constant approximation ratio. We also deal with more specialized environments. For the case of $2$ understating-only agents, we provide a randomized strategyproof $\frac{5+4\sqrt{2}}{7} \approx 1.522$-approximate mechanism, and a lower bound of $\frac{5\sqrt{5}-9}{2} \approx 1.09$. When all agents but one are honest, we provide a deterministic strategyproof $\frac{1+\sqrt{5}}{2} \approx 1.618$-approximate mechanism with a matching lower bound. Finally, we consider a model where agents can misreport their items' properties rather than existence. Specifically, each agent owns a single item, whose value-to-size ratio is publicly known, but whose actual value and size are not. We show that an adaptation of the greedy mechanism is strategyproof and $2$-approximate, and provide a matching lower bound; we also show that no deterministic strategyproof mechanism can provide a constant approximation ratio.
\end{abstract}

\section{Introduction}

In the classical knapsack problem, given a set of items, where each item has a certain size and value, the objective is to choose a subset of those items in order to maximize total value subject to a total size constraint. However, there are applications of the knapsack problem where the data regarding the items is not simply available, but rather is privately held by agents who can manipulate the data in order to advance their own goals. For example, consider scheduling MRI appointments at a hospital on a given day. Each patient has some required scan duration, and derives a certain utility from it, which represents the health benefit from her scan. The hospital would naturally like to maximize the total utility subject to the time the MRI machine is available during the day. However, the hospital consists of many doctors (or alternatively departments), each is concerned with the benefit of their own patients. Thus, if a doctor believes she can increase her patients' total utility by not ordering some necessary but less urgent scans, or by booking some unnecessary (but harmless) scans, she might do so. It may seem counterintuitive that such manipulation could be beneficial to the doctor's patients, but as we will see, it can be quite beneficial to them and quite harmful to the welfare of the other patients. It should be noted that manipulation due to selfish behavior in the healthcare system is a real problem. One example concerns transplant centers/hospitals withholding information (hiding patients) from kidney exchange programs in order to increase the number of transplants performed locally or the number of their own patients getting a kidney \cite{ashlagi2014free, hajaj2015strategy, ashlagi2013mix}. Another example is the influence of financial interests on ordering medical tests and performing medical procedures \cite{dai2012imaging, mitchell2010effect}.

Following the above discussion, we consider a variant of the knapsack problem, which we will now informally discuss (formal definitions are given in Section \ref{sec:model}). In our variant, the items are owned by a set of agents (every item is owned by a unique agent, but an agent may own multiple items). A mechanism designer needs to choose a solution, namely a subset of the items with total size less than or equal the knapsack's capacity. The designer would ``like" to choose a solution which maximizes the total value. However, each agent's set of items is private information known only to that agent, and reported to the designer by her. Furthermore, each agent is a selfish utility maximizer, and gets a utility equal to the value of her own items chosen for inclusion in the knapsack, that is the total value of the intersection between her (true) set of items and the solution. To maximize her utility, each agent can misreport her set of items. Specifically, an agent can avoid reporting some of her existing items, as well as reporting fake items which she does not actually own (she will derive no direct value from a fake item being chosen, but reporting it might increase the total value of her chosen true items).\footnote{In most of our paper, we take a somewhat untraditional approach by having the manipulated information being the existence of the items rather than their properties, such as values and sizes. Nevertheless, this approach exists in other problems as well. For example, the manipulation done by transplant centers in kidney exchange problems \cite{ashlagi2013mix} is of this type; another example can be found in \cite{dughmi2010truthful}. In the particular example of MRI, doctors not being able to manipulate the value of a scan might stem, for example, from the value being unrelated to medical considerations, and simply represent the time the agent has been waiting for a scan. Alternatively, it could stem from having a central authority that determines the medical value of a scan, but---due to the large number of patients---does not evaluate patients in advance, but rather relies on the doctors to report the need for a scan. The latter is particularly likely if the machine is externally and publicly owned, and receives requests for scans from private hospitals.} Due to this, we cannot simply assume the agents will report truthfully. Consider the following example:

\begin{example}\label{ex:abcd}
Consider the case where the mechanism designer always chooses the optimal solution to the knapsack problem based on the reports. Define the following instances of our problem. In all these instances, the knapsack's capacity is $1$, and there are two agents. The items in our instances will be taken from the following four: items $a$, $b$ and $c$, each of value $\frac{3}{4}$ and size $\frac{1}{2}$, and item $d$, of value and size $1$.

\begin{enumerate}
\item \underline{Instance 1}: agent $1$ owns item $a$, agent $2$ owns items $c$ and $d$. The optimal solution is to include items $a$ and $c$ in the knapsack.
\item \underline{Instance 2}: agent $1$ owns item $a$, agent $2$ owns item $d$. The optimal solution is to include item $d$ in the knapsack.
\item \underline{Instance 3}: agent $1$ owns items $a$ and $b$, agent $2$ owns item $d$. The optimal solution is to include items $a$ and $b$ in the knapsack.
\end{enumerate}

If the true instance is 1, agent $2$ has an incentive to not report item $c$. In this situation, the optimal solution to instance 2 will be chosen, which increases agent $2$'s utility from $\frac{3}{4}$ to $1$. If, on the other hand, the true instance is 2, agent $1$ has an incentive to report that she also owns item $b$. In this case, the solution to instance 3 would be chosen. While agent $1$ will derive no benefit from the inclusion of item $b$ in the solution, she will derive benefit from the inclusion of item $a$ in it, increasing her utility from $0$ to $\frac{3}{4}$.
\end{example}

Due to the strategic nature of the problem, we will be looking for mechanisms that achieve ``good" results, in some sense, in game-theoretic equilibrium. We will be considering two kinds of equilibria: dominant strategy equilibrium, where no agent can benefit from reporting differently regardless of the reports of the other agents, and Bayes-Nash equilibrium (BNE), where no agent can benefit from reporting differently if no other agent reports differently. Generally, achieving optimality at equilibrium will not be possible, and so we will be looking for mechanisms that approximate optimality well. Specifically, we will be looking for mechanisms that, at equilibrium, provide a small worst-case approximation ratio.

Our paper is part of a growing literature on the subject of approximate mechanism design without money, which originated in a paper by Procaccia and Tennenholtz \cite{DBLP:journals/teco/ProcacciaT13}. This approach has been applied to many types of problems, such as matching \cite{dughmi2010truthful}, facility location \cite{DBLP:journals/mor/AlonFPT10, feldman2013strategyproof, DBLP:journals/corr/abs-1305-2446}, cake cutting \cite{maya2012incentive} and kidney exchange \cite{ashlagi2013mix}. The most relevant paper we could find is by Chen et al. \cite{DBLP:journals/corr/abs-1104-2872}. In (part of) their paper, they consider what we call in our paper the ``understating model" (where agents cannot report non-existent items), but there is no overlap in their results and ours. They provide a randomized strategyproof mechanism in that model, but at a great cost to the quality of approximation. Also related is the ``Funding Games'' model of Bar-Noy et al. \cite{bar2012funding}. In that problem, the objective is again to maximize the total value of the chosen items, but the agents wish to maximize the {\em size} of their chosen items. They provide a 2-approximate mechanism for a setting in which each agent reports {\em one} of her items, specifically its true size and a lower bound on its value.

\vskip .1cm\noindent{\bf Contributions.} We provide positive and negative results for three models: the understating model---where agents can hide items but cannot report fake ones, the overstating model---where agents can report fake items but cannot hide items, and the full model---where agents can manipulate the data both by hiding items and reporting fake ones. We also consider two specialized environments: a duopoly, i.e. when the number of agents is $2$, and the one-bad-apple environment, where all agents but one are assumed to be honest. In terms of BNE, we analyze a mechanism called HALF-GREEDY and show that it \emph{provides a correlated price of anarchy of $2$ in the understating model}; this remains the case for the full model, under a mild assumption. In terms of dominant strategy equilibria, we summarize our results in Table \ref{tbl:bounds}.

\begin{table}[t!]
\begin{center}
\caption{\small Approximation Ratio Bounds for Strategyproof Mechanisms}
\begin{tabular}{ c | c | c | c || c }
  & \small{Overstating} & \small{Understat.} ($n=2$) & \small{Overstat. (1-bad-apple)} & \small{KQUS}\\\hline
\small{Rand. UB}  & $2$ &  $\frac{5+4\sqrt{2}}{7} \approx 1.522$ &$\frac{1+\sqrt{5}}{2} \approx 1.618$ & $2$\\
\small{Rand. LB}  & $2$ & $\frac{5\sqrt{5}-9}{2} \approx 1.09$ &$\frac{5\sqrt{5}-9}{2} \approx 1.09$ & $2$\\
\small{Det. UB} & -- & -- & $\frac{1+\sqrt{5}}{2} \approx 1.618$ & --\\
\small{Det. LB}  & $\infty$ & $\frac{1+\sqrt{5}}{2} \approx 1.618$ & $\frac{1+\sqrt{5}}{2} \approx 1.618$ & $\infty$
\end{tabular}\label{tbl:bounds}
\end{center}
\end{table}

It is worth noting that our positive result for the overstating model is also due to the HALF-GREEDY mechanism. We also show that HALF-GREEDY can be modified to be strategyproof and $2$-approximate in a different problem, 'Known Quality Unknown Size' (KQUS), where each agent owns a single item whose value-to-size ratio is publicly known, but its actual value and size are not.

The rest of our paper is organized as follows. In Section \ref{sec:model} we introduce our model and formalize the above discussion. In Section \ref{sec:halfgreedy}, we discuss HALF-GREEDY and its strategic properties. Sections \ref{sec:equtility} and \ref{sec:pacifytheliar} are dedicated to the specialized environments of a duopoly and one-bad-apple, respectively. In Section \ref{sec:lowerbounds}, we provide some lower bounds on the approximation ratios strategyproof mechanisms can accomplish; among other things, these lower bounds imply the optimality of HALF-GREEDY in the overstating model, and the optimality of our specialized mechanism for the case where all but one agent are honest. In Section \ref{sec:kqus}, we consider KQUS. We show that an adaptation of HALF-GREEDY is strategyproof and $2$-approximate there, and provide a matching lower bound; we also show that no deterministic strategyproof mechanism can provide a constant approximation ratio. Finally, in Section \ref{sec:future} we propose directions for future research. In addition, we provide some results and proofs omitted from the paper in the appendix.

\section{Model}\label{sec:model}

Let $C \in \mathbb{R}_{+}$ be the knapsack's capacity. Let $N=\{1,2,\ldots,n\}$ be a set of agents, $n \geq 2$. Each agent $i$ has a ground set of items $G_i$ (informally, the set of items agent $i$ can \emph{potentially} own). Each item $a \in G_i$ has size $s(a)$ and value $v(a)$ where $s(a) \in (0,C]$ and $v(a) \in (0,\infty)$.\footnote{The fact that our formulation allows us to distinguish between items with identical size, value and owner is a mere convenience. All of our results translate to a model where such items are indistinguishable (in such a model, when reporting fake items is allowed, we assume that real items are chosen before fake ones in HALF-GREEDY; this is justified by the view that the mechanism designer's choice in this model is the number of spots $q_{i,s,v}$ to allocate to agent $i \in N$ for items of size $s$ and value $v$, while agent $i$ chooses which ``specific" items to put in those spots).} We assume that $G_i \cap G_j=\varnothing$ for $i \neq j$. In the special case where for every $s \in (0,C]$, $v \in (0,\infty)$, $G_i$ contains infinitely many items of size $s$ and value $v$ (for each $i \in N$), we call the ground sets {\it unrestricted}. For technical convenience, we assume the existence of a given total order\footnote{It is easily seen that our results can be rewritten without being given this order in advance, but having it simplifies presentation as it eliminates the need for tie-breaking procedures which essentially decide such an order.} $\succeq$ on $\cup_{i \in N}G_i$. For a set $A$, let us denote the collection of all finite subsets of $A$ as $\widehat{A}$. Every agent $i$ has a finite true set of items $X_i \in \widehat{G_i}$; also, let ${\bf X}=(X_1,\ldots,X_n)$ be denoted as the true set profile. For each agent $i$, let $R_i^*(X_i) \subseteq \widehat{G_i}$ be her report space when her true set of items is $X_i$. Unless otherwise stated, we assume $R_i^*(X_i)=2^{X_i}$, a case we call the understating model (where agents can hide items but cannot report fake ones); however, we will also consider the overstating model, where $R_i^*(X_i)=\{A \in \widehat{G_i}:X_i \subseteq A\}$ (agents can report fake items but cannot hide items), and the full model, where $R_i^*(X_i)=\widehat{G_i}$ (agents can perform both types of manipulations and hence are free to report any set of items they want). Also, we call agent $i$ {\it honest} if she is always truthful, that is $R_i^*(X_i)=\{X_i\}$. Thus, for example, if we say that a result applies in the understating model with agents $i,j \in N$ assumed honest, then we mean that the result applies when $R_i^*(X_i)=\{X_i\}$, $R_j^*(X_j)=\{X_j\}$, and for $k \in N \backslash \{i,j\}$, $R_k^*(X_k)=2^{X_k}$.

A deterministic mechanism is a function $f:\widehat{G_1} \times \cdots \times \widehat{G_n} \rightarrow \widehat{\cup_{i \in N}G_i}$ which maps the agents' reports to a set of items to include in the knapsack. A randomized mechanism is a function from $\widehat{G_1} \times \cdots \times \widehat{G_n}$ to all random variables with support in $\widehat{\cup_{i \in N}G_i}$. For any $A \in \widehat{\cup_{i \in N}G_i}$, we define $s(A)=\sum_{a \in A}s(a)$ and $v(A)=\sum_{a \in A}v(a)$. Next, we define:

\begin{definition}
A deterministic (respectively, randomized) mechanism $f$ is feasible iff for all ${\bf R} \in \widehat{G_1} \times \cdots \times \widehat{G_n}$:
\begin{enumerate}
\item The mechanism only uses the available items: $f({\bf R}) \in \cup_{i \in N}R_i$ (respectively, surely).
\item The mechanism doesn't violate the knapsack's capacity: $s(f({\bf R})) \leq C$ (respectively, surely).
\end{enumerate}
\end{definition}

In this paper, whenever we write `mechanism', we mean feasible mechanism. For the rest of the paper we assume without loss of generality the normalization $C=1$. We define the utility that agent $i$ derives from $S \in \widehat{\cup_{i \in N} G_i}$ when her true set is $X_i$ to be $u(X_i,S)=v(X_i \cap S)$. We also define the notation $({\bf z}_{-i},z_i')=(z_1,\ldots,z_{i-1},z_i',z_{i+1},\ldots,z_n)$ for every $i \in N$ and $n$-dimensional vector ${\bf z}$.  Next, we define strategyproofness, which means that truthful reporting is a dominant strategy for each agent (no agent can benefit from misreporting).\footnote{The revelation principle in mechanism design tells us that any result that can be achieved in dominant strategy equilibrium can also be realized via strategyproof mechanisms \cite{hartline2012approximation}, so in the search for such equilibria we can restrict ourselves to strategyproof mechanisms without loss of generality.}

\begin{definition}
A deterministic (respectively, randomized) mechanism $f$ is called {\emph strategyproof} if for all $i \in N$,  ${\bf X} \in \widehat{G_1} \times \cdots \times \widehat{G_n}$, $R_i \in R_i^*(X_i)$, we have $u(X_i,f({\bf X})) \geq u(X_i,f({\bf X}_{-i},R_i))$ (respectively, $\mathbb{E}[u(X_i,f({\bf X}))] \geq \mathbb{E}[u(X_i,f({\bf X}_{-i},R_i))]$).
\end{definition}

The social welfare obtained when the chosen solution is $S \in \widehat{\cup_{i \in N}G_i}$ and the true set profile is ${\bf X} \in \widehat{G_1} \times \cdots \times \widehat{G_n}$ is $sw(S,{\bf X})=\sum_{i=1}^n u(X_i, S)=v(S \cap (\cup_{i \in N} X_i))$. Informally speaking, the designer would like to choose $S$ that maximizes this objective function, but as we saw in Example \ref{ex:abcd}, always doing so generally violates strategyproofness. Instead, the mechanism designer's objective is to minimize the worst-case approximation ratio; for a deterministic (respectively, randomized) mechanism, it means to minimize $\max_{{\bf X} \in \widehat{G_1} \times \cdots \times \widehat{G_n}}{\frac{\sum_{i=1}^n u(X_i,OPT(\cup_{i \in N}X_i))}{\sum_{i=1}^n u(X_i, f({\bf X}))}}$ (respectively, $\max_{{\bf X} \in \widehat{G_1} \times \cdots \times \widehat{G_n}}{\frac{\sum_{i=1}^n u(X_i,OPT(\cup_{i \in N}X_i))}{\sum_{i=1}^n \mathbb{E}[u(X_i, f({\bf X}))]}}$), where $OPT(A)$ is an optimal solution to the knapsack problem when the set of available items is $A$ (in a non-strategic environment).\footnote{In case there is more than one optimal solution, choose one arbitrarily.} Thus, our goal is to find deterministic and randomized mechanisms which are strategyproof, and subject to that, have as small as possible worst-case approximation ratios. If a mechanism has worst-case approximation ratio less than or equal to $\alpha$, we say that the mechanism is $\alpha$-approximate.

Finally, one of the mechanisms we provide is not strategyproof in some of our models, but still has attractive strategic properties in terms of Bayes-Nash Equilibria (BNE). Thus, we state the required definitions here. Assume ${\bf \dot{X}}$ is a random variable over $\widehat{G_1} \times \cdots \times \widehat{G_n}$ with probability distribution $\mathcal{F}$. A strategy $\tilde{R}_i$ of agent $i$ is a function from $\widehat{G_i}$ to the space of all random variables over $\widehat{G_i}$, such that $\tilde{R}_i(\dot{X}_i) \in R_i^*(\dot{X}_i)$ surely. For convenience of notation, we make the dependence on $\dot{X}_i$ implicit and define agent $i$'s report to be $\dot{R}_i=\tilde{R}_i(\dot{X}_i)$. We say that a strategy profile ${\bf \tilde{R}}$ is a BNE under mechanism $f$ and probability distribution $\mathcal{F}$ if no agent has a beneficial unilateral deviation, meaning that there does not exist agent $i \in N$ and strategy $\tilde{R}_i'$ such that $\mathbb{E}[u(\dot{X}_i,f({\bf \dot{R}}))]<\mathbb{E}[u(\dot{X}_i,f({\bf \dot{R}}_{-i},\dot{R}_i'))]$.\footnote{Note that there is no need to separate the definitions of deterministic and randomized mechanisms in this case, as expectation is taken in both cases---in the randomized case there is just an extra random variable determined by $f$.} We say that a BNE ${\bf \tilde{R}}$ is $\alpha$-approximate if $\frac{\sum_{i=1}^n\mathbb{E}[u(\dot{X}_i,OPT(\cup_{i \in N}\dot{X}_i))]}{\sum_{i=1}^n\mathbb{E}[u(\dot{X}_i,f({\bf \dot{R}}))]} \leq \alpha$. If, for a given mechanism $f$, for all probability distributions $\mathcal{F}$ over $\widehat{G_1} \times \cdots \times \widehat{G_n}$, every BNE is $\alpha$-approximate, we say that $f$ has a {\it correlated price of anarchy} $\alpha$ (as defined in \cite{roughgarden2015price}).

\section{The HALF-GREEDY Mechanism}\label{sec:halfgreedy}

In this section, we analyze the strategic properties of a randomized mechanism we call HALF-GREEDY. In the overstating model, we show that HALF-GREEDY is strategyproof and $2$-approximate. In the understating model, we show that while HALF-GREEDY is not strategyproof, it has a correlated price of anarchy of $2$. In the full model, we can preserve a correlated price of anarchy of $2$, under a mild additional assumption. In Section 6, we obtain some lower bounds for the overstating model, which imply the optimality of HALF-GREEDY among randomized strategyproof mechanisms in that model, and also show that randomization is necessary for obtaining any constant worst-case approximation ratio.

To define HALF-GREEDY, we need two auxiliary mechanisms. The first one is the GREEDY mechanism, which adds items to the knapsack by decreasing value-to-size ratio, breaking ties according to $\succeq$. It will be convenient to define $\succeq'$ to be a total order on $\cup_{i \in N} G_i$, where for $a,b \in \cup_{i \in N} G_i$, if $\frac{v(a)}{s(a)}>\frac{v(b)}{s(b)}$, then $a \succ' b$, and if $\frac{v(a)}{s(a)}=\frac{v(b)}{s(b)}$, then $\succeq'$ agrees with $\succeq$.

\begin{definition}
For every $A \in \widehat{\cup_{i \in N} G_i}$, and every $b \in A$, let the set of items (in $A$) strictly larger than $b$ under $\succeq'$ be defined as $L_A(b)=\{a \in A:a \succ' b\}$; if $s(A) >1$, let the first item left out of the knapsack be defined as $o_A=\max_{\succeq'}\{b \in A: s(L_A(b)\cup \{b\})>1\}$. Let the reported items be ${\bf R} \in \widehat{G_1} \times \cdots \times \widehat{G_n}$. Define the GREEDY mechanism $GR$ as follows: if $s(\cup_{i \in N} R_i) \leq 1$, return $\cup_{i \in N} R_i$, otherwise return $L_{\cup_{i \in N} R_i}(o_{\cup_{i \in N} R_i})$.
\end{definition}

The second auxiliary mechanism is MAXIMUM-VALUE, which returns a single item with the maximum value possible, breaking ties according to $\succeq$:

\begin{definition}
Let the reported items be ${\bf R} \in \widehat{G_1} \times \cdots \times \widehat{G_n}$. Define the MAXIMUM-VALUE mechanism $MV$ as follows: if $\cup_{i \in N} R_i=\varnothing$ return $\varnothing$, and otherwise return $\max_{\succeq}\{a \in \cup_{i \in N} R_i: v(a) \geq v(b) \;\; \forall b \in \cup_{i \in N} R_i\}$.
\end{definition}

Now we can define HALF-GREEDY:

\begin{definition}
The HALF-GREEDY mechanism $HG$ runs GREEDY with probability $\frac{1}{2}$ and MAXIMUM-VALUE with probability $\frac{1}{2}$ (the probabilities are chosen independently of the input).
\end{definition}

GREEDY and MAXIMUM-VALUE are trivially feasible, and hence so is HALF-GREEDY. It is well known that HALF-GREEDY is $2$-approximate (this is proven via a simple linear programming based argument, see \cite{burke2005search}).

For every set $A \in \widehat{\cup_{i \in N}G_i}$, we define ${\bf X}^A$ to be the unique set profile where $A=\cup_{i \in N} X_i^A$. The following lemma shall be useful for our purposes:

\begin{lemma}\label{lem:remains}
Let $A,B,C \in \widehat{\cup_{i \in N}G_i}$, and assume $B \subseteq A$, $C \cap A \subseteq B$. Then,  $C \cap HG({\bf X}^A) \subseteq HG({\bf X}^B)$ surely (that is, every item in $C$ that is included in the knapsack when we run $MV$/$GR$ on ${\bf X}^A$ remains in the knapsack when we run them on ${\bf X}^B$).
\end{lemma}

\begin{proof}
Let $c \in C \cap MV({\bf X}^A)$; note that by feasibility of $MV$, $c \in A$ and hence $c \in C \cap A$, implying $c \in B$. Now, by definition of $MV$, we have that $c=\max_{\succeq}\{a \in A: v(a) \geq v(b) \; \; \forall b \in A\}$. Since $B \subseteq A$ and $c \in B$, we have that $c=\max_{\succeq}\{a \in B: v(a) \geq v(b) \; \; \forall b \in B\}$ (as $c$ is maximal in the larger set, it remains maximal in the smaller set), so $c \in MV({\bf X}^B)$.

Now, let $d \in C \cap GR({\bf X}^A)$; by feasibility of $GR$, $d \in A$, hence $d \in C \cap A$, hence $d \in B$. If $s(B) \leq 1$, then $GR({\bf X}^B)=B$, and in particular $d \in GR({\bf X}^B)$. So assume $s(B)>1$, and note that this implies $s(A)>1$. Clearly, $o_A \succeq' o_B$; since $d \in GR({\bf X}^A)=L_A(o_A)$, $d \succ' o_A \succeq' o_B$. Since $d \in B$ and $d \succ' o_B$, then $d \in L_B(o_B)$, so $d \in GR({\bf X}^B)$.
\end{proof}

From this lemma, we can deduce two corollaries. The first corollary states that if an agent hides items in HALF-GREEDY, she does not hurt the other agents; all the items of the other agents that were included in the knapsack before she hid items, remain included in the knapsack after.

\begin{corollary}\label{lem:doesnthurt}
Let $i \in N$, $Y_i, Z_i \in \widehat{G_i}$, $Y_i \subseteq Z_i$. Then for every $(R_1,\ldots,R_{i-1},R_{i+1},\ldots,R_n)$, $(X_1,\ldots,X_{i-1},X_{i+1},\ldots,X_n) \in \widehat{G_1} \times \cdots \times \widehat{G_{i-1}} \times \widehat{G_{i+1}} \times \cdots \times \widehat{G_n}$, for all $j \in N \backslash \{i\}$, $X_j \cap HG({\bf R}_{-i},Z_i) \subseteq HG({\bf R}_{-i},Y_i)$ surely (and thus $\mathbb{E}[u(X_j,HG({\bf R}_{-i},Y_i))] \geq \mathbb{E}[u(X_j,HG({\bf R}_{-i},Z_i))]$).
\end{corollary}

\begin{proof}
Apply Lemma \ref{lem:remains} with $C=X_j$, $B=Y_i \cup (\cup_{k \in N \backslash \{i\}} R_k)$ and $A=Z_i \cup (\cup_{k \in N \backslash \{i\}} R_k)$.
\end{proof}

The second corollary states that an agent cannot benefit from reporting fake items:

\begin{corollary}\label{cor:impliesstratproof}
Let $i \in N$, ${\bf X} \in \widehat{G_1} \times \cdots \times \widehat{G_n}$ and $R_i \in \widehat{G_i}$.Then $X_i \cap HG({\bf X}_{-i},R_i) \subseteq HG({\bf X}_{-i},R_i \cap X_i)$ surely (and thus $\mathbb{E}[u(X_i,HG({\bf X}_{-i},R_i \cap X_i))] \geq \mathbb{E}[u(X_i,HG({\bf X}_{-i},R_i))]$).
\end{corollary}

\begin{proof}
Apply Lemma \ref{lem:remains} with $C=X_i$, $B=(R_i \cap X_i) \cup (\cup_{j \in N \backslash \{i\}} X_j)$ and $A=R_i \cup (\cup_{j \in N \backslash\{i\}} X_j)$.\end{proof}

Next, we will use the corollaries above to analyze the strategic properties of HALF-GREEDY. We begin by noting that Corollary \ref{cor:impliesstratproof} guarantees strategyproofness in the overstating model:

\begin{corollary}
In the overstating model, $HG$ is strategyproof and $2$-approximate.
\end{corollary}

\begin{proof}
Strategyproofness is immediate from Corollary \ref{cor:impliesstratproof}, since for every agent $i$ and for every possible report $R_i \supseteq X_i$ in the overstating model, $R_i \cap X_i=X_i$. The fact that the mechanism is $2$-approximate is already known, as noted above.
\end{proof}

It is important to note that once the GREEDY mechanism first fails in adding an item to the knapsack, it stops and returns the items currently in the knapsack; it does \emph{not} try to add to the knapsack the next item in the ordering that fits in the remaining space. This seemingly trivial choice is actually crucial for maintaining strategyproofness, as the example below shows. Thus, one must be careful about choices that are seemingly unimportant for optimization, as they may be important for strategic properties.

\begin{example}
Consider the mechanism BAD-GREEDY $BG$, defined as Algorithm \ref{alg:badgreedy}. Consider the case of $n=2$, with $X_1=\{a\}$, $X_2=\{b\}$, $v(a)=s(a)=1$, $v(b)=\frac{1}{4}$, $s(a)=\frac{1}{2}$; on this instance, $BG({\bf X})=\{a\}$, and the utility of agent $2$ is $v(X_2 \cap \{a\})=0$. However, if agent $2$ reports $R_2=\{b,c\}$, where $v(c)=1$, $s(c)=\frac{1}{2}$ (that is, agent $2$ reports a fake item $c$ in addition to her true item $b$), then $BG(X_1,R_2)=\{b,c\}$, and agent $2$'s utility is $v(X_2 \cap \{b,c\})=\frac{1}{4}$. Thus, BAD-GREEDY is not strategyproof in the overstating model (assuming $a \in G_1$, $b,c \in G_2$).

\begin{algorithm}[t]
\SetAlgoNoLine
\KwIn{a set $A \in \widehat{\cup_{i \in N}G_i}$}
$S \leftarrow \varnothing$, $T \leftarrow A$

\While{$T \neq \varnothing$}{
$next \leftarrow \max_{\succeq'}T$

\If{$s(S \cup \{next\}) \leq 1$}{
$S \leftarrow S \cup \{next\}$}}

\Return{$S$}
\caption{BAD-GREEDY}\label{alg:badgreedy}
\end{algorithm}
\end{example}

When agents can hide items $HG$ is no longer strategyproof, but it still has attractive strategic properties in terms of Bayes-Nash equilibria.\footnote{For example, consider the case of $n=2$, where $X_1=\{a,b\}$, $X_2=\{c,d\}$, $v(a)=2$, $v(c)=2-\epsilon$, $s(a)=s(c)=\frac{1}{4}+\epsilon$, $v(b)=3-\epsilon$, $v(d)=3$, $s(b)=s(d)=\frac{1}{2}$, where $\epsilon>0$ is very small. It is easy to check that there are no dominant strategies in this case, truthful or not.} Note that $HG$ is prior independent, meaning that the designer does not need to know the probability distribution of ${\bf \dot{X}}$ in order to run it. Let us look into its performance in the understating model:

\begin{theorem}\label{halfgreedycorrpoa2}
In the understating model, HALF-GREEDY has a correlated price of anarchy of $2$.
\end{theorem}

\begin{proof}
Fix agent $i \in N$. Let ${\bf \dot{X}}$ be a random variable over $\widehat{G_1} \times \cdots \times \widehat{G_n}$ with probability distribution $\mathcal{F}$. Let ${\bf \tilde{R}}$ be a BNE under mechanism $HG$ and probability distribution $\mathcal{F}$. Note that since we are in the understating model, $\dot{R}_j \subseteq \dot{X}_j$ surely for every $j \in N$  (reminder: $\dot{R}_j=\tilde{R}_j(\dot{X}_j)$). Thus, we can apply Corollary \ref{lem:doesnthurt} $n-1$ times to deduce $\mathbb{E}[u(\dot{X}_i,HG({\bf \dot{X}}))] \leq \mathbb{E}[u(\dot{X}_i,HG({\bf \dot{R}}_{-i},\dot{X}_i))]$; that is, when all agents other than $i$ hide items, agent $i$'s utility weakly increases. Now, by definition of BNE, $\mathbb{E}[u(\dot{X}_i,HG({\bf \dot{R}}_{-i},\dot{X}_i))] \leq \mathbb{E}[u(\dot{X}_i,HG({\bf \dot{R}}))]$ (when all agents other than $i$ play according to ${\bf \tilde{R}}$, $i$'s optimal response is to play $\tilde{R}_i$). Thus, we have that $\mathbb{E}[u(\dot{X}_i,HG({\bf \dot{X}}))] \leq \mathbb{E}[u(\dot{X}_i,HG({\bf \dot{R}}))]$. As $i$ was chosen arbitrarily, this holds for all agents, and so we have that $\sum_{i=1}^n\mathbb{E}[u(\dot{X}_i,HG({\bf \dot{X}}))] \leq \sum_{i=1}^n\mathbb{E}[u(\dot{X}_i,HG({\bf \dot{R}}))]$.

Now, as we noted, $HG$ is $2$-approximate, which means that $\frac{\sum_{i=1}^n u(\dot{X}_i,OPT(\cup_{j \in N} \dot{X}_j))}{\sum_{i=1}^n \mathbb{E}[u(\dot{X}_i,HG({\bf \dot{X}}))]} \leq 2$ for any fixed ${\bf \dot{X}}$, which implies $\frac{\sum_{i=1}^n \mathbb{E}[u(\dot{X}_i,OPT(\cup_{j \in N} \dot{X}_j))]}{\sum_{i=1}^n \mathbb{E}[u(\dot{X}_i,HG({\bf \dot{X}}))]} \leq 2$ (in the latter case, the expectation is taken over ${\bf \dot{X}}$ as well). Thus, we have that $\frac{\sum_{i=1}^n \mathbb{E}[u(\dot{X}_i,OPT(\cup_{j \in N} \dot{X}_j))]}{\sum_{i=1}^n \mathbb{E}[u(\dot{X}_i,HG({\bf \dot{R}}))]} \leq 2$ as well.
\end{proof}

Next, we consider the full model. In the full model we ``almost" have a correlated price of anarchy of $2$. The reason we say ``almost" is indifference: in BNE, an agent might report fake items in a way that does not change her utility, but decreases other agents' utilities. Let us give an example of such a Nash equilibrium, which is a special case of BNE when there is no private data (that is, when each $\dot{X}_i$ equals a certain given set of items surely):

\begin{example}
Consider the case of $n=2$ agents, $X_1=\{a\}$, $X_2=\{b\}$, where $v(a)=1$, $s(a)=\frac{1}{M}$, $v(b)=M-2$, $s(b)=\frac{M-1}{M}$, where $M$ is some large integer, $M>>2$. Truthful reporting is a Nash equilibrium. Note that when agents report truthfully, $HG$ chooses $a$ with probability $\frac{1}{2}$ and $b$ with probability $1$. However, if agent 1 reports $R_1=\{a,c\}$ where $v(c)=M-1$ and $s(c)=\frac{M-1}{M}$, and agent $2$ reports truthfully, we still get a Nash equilibrium, in which $HG$ still chooses $a$ with probability $\frac{1}{2}$, but $b$ is chosen with probability $0$ ($c$ is chosen with probability $1$, but since it is a fake item, it does not add to the agents' utilities or to the objective function value). In the latter Nash equilibrium, the approximation ratio is $2M-2$.
\end{example}

Corollary \ref{cor:impliesstratproof} tells us that an agent can never benefit from reporting fake items, so in a way the above can be viewed as a technicality. We shall get around that technicality by assuming that agents are not malicious.

\begin{definition}
Let $i \in N$, $X_i,R_i \in \widehat{G_i}$. $R_i$ is called a {\it malicious report} for agent $i$ with true set of items $X_i$ if there exists $R_i' \in \widehat{G_i}$ where for all $(X_1,\ldots,X_{i-1},X_{i+1},\ldots,X_n) \in \widehat{G_1}\times\cdots\times\widehat{G_{i-1}}\times\widehat{G_{i+1}}\times\cdots\times\widehat{G_n}$, and all $j \in N$, $\mathbb{E}[u(X_j,HG({\bf X}_{-i},R_i))] \leq \mathbb{E}[u(X_j,HG({\bf X}_{-i},R_i'))]$, with the inequality being strict for at least one agent in at least one instance.

A strategy $\tilde{R}_i$ is called malicious if there exists $Y_i \in \widehat{G_i}$ so that the probability that $\tilde{R}_i(Y_i)$ is malicious (for $i$ when her true set of items is $Y_i$) is nonzero.
\end{definition}

In other words, a malicious report is a report that can never benefit any agent (including the agent reporting it), and can sometimes hurt an agent. Thus, if the agents are even very mildly altruistic, they would not report maliciously.\footnote{If we want, we can instead assume that an agent's objective is lexicographic, where she first attempts to maximize the value of her own items in the knapsack, and then, subject to that, also tries to maximize social welfare; in that case, a malicious strategy is a weakly dominated strategy. Thus, not reporting maliciously is implied by not playing weakly dominated strategies, so in this model we can replace the assumption of no malicious reporting by no weakly dominated strategies being played and still get a correlated price of anarchy of $2$.} Next, we show an important property of non-malicious reports: fake items included in those reports have no impact on the true items included in the solution. The following lemma proves a special case of this claim, which limits all agents but one to report truthfully; the general claim is then proven as a corollary.

\begin{lemma}\label{lem:notmal}
Let $i \in N$, $X_i,R_i \in \widehat{G_i}$. If $R_i$ is not malicious for agent $i$ with true set of items $X_i$, then for every choice of $(X_1,\ldots,X_{i-1},X_{i+1},\ldots,X_n) \in \widehat{G_1}\times\cdots\times\widehat{G_{i-1}}\times\widehat{G_{i+1}}\times\cdots\times\widehat{G_n}$, and for every $j \in N$, we have that $X_j \cap HG({\bf X}_{-i},R_i)=X_j \cap HG({\bf X}_{-i},R_i \cap X_i)$ surely.
\end{lemma}

\begin{proof}
Assume that there exists  $(X_1,\ldots,X_{i-1},X_{i+1},\ldots,X_n) \in \widehat{G_1}\times\cdots\times\widehat{G_{i-1}}\times\widehat{G_{i+1}}\times\cdots\times\widehat{G_n}$ such that $X_j \cap MV({\bf X}_{-i},R_i) \neq X_j \cap MV({\bf X}_{-i},R_i \cap X_i)$ for some $j \in N$ (the proof for $GR$ instead of $MV$ is identical). 

First, note that for every instance where the true set of items for agent $i$ is $X_i$, and for any fixed reports for the agents in $N \backslash \{i\}$, reporting $R_i \cap X_i$ instead of $R_i$ weakly increases all agents' utilities (agent $i$ by Corollary \ref{cor:impliesstratproof}, and the rest by Corollary \ref{lem:doesnthurt}). Thus, it is enough to show that on our instance, when all $k \in N \backslash \{i\}$ own and report $X_k$, reporting $R_i \cap X_i$ instead of $R_i$ strictly increases some agent's utility (since then $R_i$ would be malicious). Corollaries \ref{cor:impliesstratproof} and \ref{lem:doesnthurt} imply that $u(X_k,MV({\bf X}_{-i},R_i \cap X_i)) \geq u(X_k,MV({\bf X}_{-i},R_i))$ and $u(X_k,GR({\bf X}_{-i},R_i \cap X_i)) \geq u(X_k,GR({\bf X}_{-i},R_i))$ for all $k \in N$. By Corollary \ref{lem:doesnthurt} if $j \neq i$ and Corollary \ref{cor:impliesstratproof} if $j=i$, $X_j \cap MV({\bf X}_{-i},R_i) \neq X_j \cap MV({\bf X}_{-i},R_i \cap X_i)$ implies $X_j \cap MV({\bf X}_{-i},R_i) \subset X_j \cap MV({\bf X}_{-i},R_i \cap X_i)$, thus $u(X_j,MV({\bf X}_{-i},R_i \cap X_i))>u(X_j,MV({\bf X}_{-i},R_i))$. Thus our proof is complete.
\end{proof}

\begin{corollary}\label{cor:nonmal}
Let $i \in N$, $X_i,R_i \in \widehat{G_i}$. If $R_i$ is not malicious for agent $i$ with true set of items $X_i$, then for every choice of $(X_1,\ldots,X_{i-1},X_{i+1},\ldots,X_n),(R_1,\ldots,R_{i-1},R_{i+1},\ldots,R_n) \in \widehat{G_1}\times\cdots\times\widehat{G_{i-1}}\times\widehat{G_{i+1}}\times\cdots\times\widehat{G_n}$, and for every $j \in N$, we have that $X_j \cap HG({\bf R})=X_j \cap HG({\bf R}_{-i},R_i \cap X_i)$ surely.
\end{corollary}

\begin{proof}
Let  $(R_1,\ldots,R_{i-1},R_{i+1},\ldots,R_n),(X_1,\ldots,X_{i-1},X_{i+1},\ldots,X_n) \in \widehat{G_1}\times\cdots\times\widehat{G_{i-1}}\times\widehat{G_{i+1}}\times\cdots\times\widehat{G_n}$. If $R_i$ is non-malicious, then by Lemma \ref{lem:notmal} we know that for all $j \in N$, $R_j \cap HG({\bf R})=R_j \cap HG({\bf R}_{-i},R_i \cap X_i)$ surely. Intersecting both sides with $X_j$ yields $X_j \cap R_j \cap HG({\bf R})=X_j \cap R_j \cap HG({\bf R}_{-i},R_i \cap X_i)$ surely. By feasibility of $HG$, we have $X_j \cap HG({\bf R}) \subseteq R_j$ and $X_j \cap HG({\bf R}_{-i},R_i \cap X_i) \subseteq R_j$; thus it follows that $X_j \cap HG({\bf R})=X_j \cap HG({\bf R}_{-i},R_i \cap X_i)$
\end{proof}

We can now describe the sense in which the correlated price of anarchy of $2$ is preserved:

\begin{theorem}
In the full model, any BNE under $HG$ in which no malicious strategies are played is $2$-approximate.
\end{theorem}

\begin{proof}
Let ${\bf \dot{X}}$ be a random variable over $\widehat{G_1} \times \cdots \times \widehat{G_n}$ with probability distribution $\mathcal{F}$. Let ${\bf \tilde{R}}$ be a BNE under mechanism $HG$ and probability distribution $\mathcal{F}$, and assume that no malicious strategies are being played in ${\bf \tilde{R}}$. For convenience, define ${\bf \dot{R} \cap \dot{X}}=(\dot{R}_1 \cap \dot{X}_1,\ldots,\dot{R}_n \cap \dot{X}_n)$, and let ${\bf \tilde{R} \cap \dot{X}}$ the strategy profile generating those reports. Applying Corollary \ref{cor:nonmal} $n$ times in succession, once for each agent, to get from ${\bf \dot{R}}$ to ${\bf \dot{R} \cap \dot{X}}$, we can immediately conclude that every agent gets the exact same utility under ${\bf \dot{R}}$ and ${\bf \dot{R} \cap \dot{X}}$. If ${\bf \tilde{R} \cap \dot{X}}$ is a BNE in the full model, then since no fake items are reported, it is clearly also a BNE in the understating model, and hence by Theorem \ref{halfgreedycorrpoa2} it is $2$-approximate, and therefore ${\bf \tilde{R}}$ is $2$-approximate as well. Thus, it is enough to show that ${\bf \tilde{R} \cap \dot{X}}$ is a BNE in the full model.

Fix agent $i$. Apply Corollary \ref{cor:nonmal} $n-1$ times to get that for all $Y_i \in \widehat{G_i}$, $\mathbb{E}[u(\dot{X}_i,HG({\bf \dot{R}_{-i}},Y_i))]=\mathbb{E}[u(\dot{X}_i,HG({\bf \dot{R}\cap\dot{X}}_{-i},Y_i))]$; thus, it follows that since $\dot{R}_i$ is a best response by agent $i$ when all other agents $j$ play $\dot{R}_j$, it remains a best response when they all play $\dot{R}_j \cap \dot{X}_j$ instead. Also, by Corollary \ref{cor:impliesstratproof}, agent $i$ weakly benefits from playing $\dot{R}_i \cap \dot{X}_i$ instead of $\dot{R}_i$, so $\dot{R}_i \cap \dot{X}_i$ is also a best response to all other agents $j$ playing $\dot{R}_j \cap \dot{X}_j$. Thus, ${\bf \tilde{R}\cap\dot{X}}$ is a BNE in the full model.
\end{proof}

As we noted in the proof above, a BNE in the full model in which no fake items are reported is trivially a BNE in the understating model. However, it is not trivial that every BNE in the understating model remains a BNE in the full model. Below we show that it is in fact the case, so no BNE are lost when moving from the understating model to the full model.

\begin{theorem}\label{thm:bnefullmodel}
Let ${\bf \tilde{R}}$ be a BNE in under $HG$ and distribution $\mathcal{F}$ in the understating model; then it remains such a BNE in the full model.
\end{theorem}

\begin{proof}
Assume ${\bf \tilde{R}}$ is not a BNE in the full model. Then there must exist an agent $i$ and strategy $\tilde{R}_i'$ where $\mathbb{E}[u(\dot{X}_i,HG({\bf \dot{R}}))]<\mathbb{E}[u(\dot{X}_i,HG({\bf \dot{R}}_{-i},\dot{R}_i'))]$. But, by Corollary \ref{cor:impliesstratproof}, this implies that $\mathbb{E}[u(\dot{X}_i,HG({\bf \dot{R}}))]<\mathbb{E}[u(\dot{X}_i,HG({\bf \dot{R}}_{-i},\dot{R}_i'\cap \dot{X}_i))]$, and since $\dot{R}_i'\cap \dot{X}_i \subseteq \dot{X}_i$ surely, it follows that ${\bf \tilde{R}}$ is not a BNE in the understating model.
\end{proof}

To conclude our discussion on BNE in this section, we note that if we assume that $G_1,\ldots,G_n$ are finite, the existence of a ($2$-approximate) BNE in the understating model (and hence, by Theorem \ref{thm:bnefullmodel}, also in the full model) is guaranteed by Nash's Theorem, and all of our results so far still hold.

\section{The EQUAL-UTILITY Mechanism}\label{sec:equtility}

In this section, we consider the special case of $n=2$ in the understating model. We design a specialized randomized mechanism for this environment, which is strategyproof and $\frac{5+4\sqrt{2}}{7} \approx 1.522$-approximate. We note that in Section 6, we provide a lower bound of $\frac{5\sqrt{5}-9}{2} \approx 1.09$ for randomized strategyproof mechanisms in the understating model, even when $n=2$; thus, if we want strategyproofness, some approximation gap is necessary. Also in Section 6, we show that no deterministic strategyproof mechanism can beat our mechanism's approximation ratio.

Informally, the idea behind EQUAL-UTILITY, shown as Algorithm \ref{alg:equtil}, is to consider $OPT(X_1)$ and $OPT(X_2)$, namely the optimal solutions using just a single agent's items. Let's say we wish for EQUAL-UTILITY to be $\alpha$-approximate. If $OPT(X_i)$ is significantly bigger than $OPT(X_j)$, to the extent where $OPT(X_i)$ is guaranteed to be an $\alpha$-approximation on its own to the optimal value, then we simply return $OPT(X_i)$. Otherwise, we have some indication that the two agents' ``importance" in creating approximately optimal solutions is relatively close. In that case, we consider the following mathematical program (PROGRAM), where $A$ is a random decision variable (set of items): maximize $v(A)$, subject to (1) $A \subseteq X_1 \cup X_2$ and $s(A) \leq 1$ surely and (2) $\mathbb{E}[v(A \cap X_1)]=\mathbb{E}[v(A \cap X_2)]$.\footnote{PROGRAM can be stated as a linear programming problem with exponentially many variables. Let $T=\{S \subseteq X_1 \cup X_2:s(S) \leq 1\}$. Then PROGRAM can be stated as: maximize $\sum_{S \in T} v(S)p_S$ subject to $\sum_{S \in T} v(S \cap X_1)p_S=\sum_{S \in T} v(S \cap X_2)p_S$, $\sum_{S \in T}p_S=1$ and $p_S \geq 0$ for all $S \in T$ (where the $p_S$'s are our decision variables).}

PROGRAM solves the knapsack problem with the additional constraint that the agents' utilities are equal. Since in general, apart from $\varnothing$, there might not be a deterministic solution that satisfies this additional constraint, we allow for randomized solutions to the knapsack problem instead. The important property of PROGRAM is that it aligns the strategic goals of the two agents: each agent cares about the other's success exactly as much as she does about her own, for their utilities are equal. The fact that the agents' ``importance" is close makes the equal expected utility constraint not too restricting, and causes this program to provide a good approximation to the optimal value.

\begin{algorithm}[t]
\SetAlgoNoLine
\KwIn{Sets of items $X_1,X_2$, where $X_i \in \widehat{G_i}$; parameter $\alpha \in [1,2)$}
$Z_1 \leftarrow OPT(X_1)$, $Z_2 \leftarrow OPT(X_2)$

\eIf{$v(Z_i) \geq \frac{1}{\alpha}(v(Z_1)+v(Z_2))$ for some $i \in \{1,2\}$}{\Return $Z_i$ (option 1)}{\Return optimal solution to PROGRAM with input ${\bf X}$ (option 2)}
\caption{EQUAL-UTILITY}\label{alg:equtil}
\end{algorithm}

We have used options 1 and 2 to denote the lines in which the algorithm can terminate.  If not both $v(Z_1)=v(Z_2)=0$, note that since $\alpha<2$, there is no ambiguity in the conditional statement: if there is a choice of $i$ that satisfies the condition, it must be unique. If both $v(Z_1)=v(Z_2)=0$, then $Z_1=Z_2=\varnothing$, so which one is chosen does not matter.

\begin{theorem}\label{thm:equtilstratproof}
For $\alpha \geq \frac{5+4\sqrt{2}}{7} \approx 1.522$, EQUAL-UTILITY is strategyproof and $\alpha$-approximate.
\end{theorem}

\begin{proof}
First, we prove strategyproofness. We break the proof into cases:
\begin{enumerate}
\item Assume the mechanism ends at option 1. In that case, agent $i$ clearly gets the best utility she can possibly get, and hence has no incentive to misreport. Consider agent $j \neq i$. Assume agent $j$ reports $X_j' \subset X_j$. Let us denote $Z_j'=OPT(X_j')$. Note that $v(Z_j') \leq v(Z_j)$. Since $v(Z_i) \geq \frac{1}{\alpha}(v(Z_i)+v(Z_j))$, trivially  $v(Z_i) \geq \frac{1}{\alpha}(v(Z_i)+v(Z_j'))$. Thus, agent $j$ cannot prevent the mechanism from ending at option 1 by misreporting. Furthermore, since at option 1 the mechanism returns $Z_i$, agent $j$ has no influence on what is returned. Thus agent $j$ cannot change the outcome of the mechanism by misreporting.
\item Assume the mechanism ends at option 2. Consider agent 1 (the proof for agent 2 is identical). Assume agent $1$ reports $X_1' \subset X_1$. Let us denote $Z_1'=OPT(X_1')$. Note that $v(Z_1') \leq v(Z_1)$. Since the mechanism did not end at option 1, we have that $v(Z_1)<\frac{1}{\alpha}(v(Z_1)+v(Z_2))$, that is $(\alpha-1)v(Z_1)<v(Z_2)$, so $(\alpha-1)v(Z_1')<v(Z_2)$ and hence $v(Z_1')<\frac{1}{\alpha}(v(Z_1')+v(Z_2))$. Thus, agent 1 cannot make the mechanism stop at option 1 and return $Z_1'$, and therefore agent 1 cannot benefit from making the mechanism stop at option 1 (since if $Z_2$ is returned, her payoff is $0$). At option 2, agent 1 would like to report all of her items: her utility is exactly half of the optimal solution to PROGRAM, and so enlarging the feasible region of PROGRAM weakly increases her own utility.
\end{enumerate}

Next, we prove the mechanism is $\alpha$-approximate. As $v(Z_1)+V(Z_2) \geq v(OPT(\cup_{i \in \{1,2\}} X_i))$, if the mechanism ends at option 1 clearly it provides an $\alpha$-approximation.\footnote{It is worth noting that it is important to use $v(Z_1)+v(Z_2)$ instead of $v(OPT(X_1 \cup X_2))$ for the conditional statement; using the latter would violate strategyproofness.} So we just need to prove this for the case the mechanism ends at option 2. Let $O=OPT(X_1 \cup X_2)$, $O_1=O \cap X_1$, $O_2=O \cap X_2$. Let $a=\argmin_{i \in \{1,2\}} v(O_i)$, $b=\argmax_{i \in \{1,2\}} v(O_i)$ (if $v(O_1)=v(O_2)$, set $a=1$ and $b=2$), and let $p=\frac{v(O_b)-v(O_a)}{v(O_b)-v(O_a)+v(Z_a)}$. Consider the random variable $A$ which returns $Z_a$ with probability $p$ and $O$ with probability $(1-p)$. $p$ was chosen precisely so that $A$ becomes a feasible solution to our program: $\mathbb{E}[v(A \cap X_b)]=(1-p)v(O_b)=\frac{v(O_b)v(Z_a)}{v(O_b)-v(O_a)+v(Z_a)}=\frac{(v(O_b)-v(O_a))v(Z_a)}{v(O_b)-v(O_a)+v(Z_a)}+\frac{v(O_a)v(Z_a)}{v(O_b)-v(O_a)+v(Z_a)}=pv(Z_a)+(1-p)v(O_a)=\mathbb{E}[v(A \cap X_a)]$. It is therefore enough to show that $\frac{v(O)}{\mathbb{E}[v(A)]} \leq \alpha$. 

Note that $\frac{v(O)}{\mathbb{E}[v(A)]}=\frac{v(O_a)+v(O_b)}{(1-p)(v(O_a)+v(O_b))+pv(Z_a)}=\frac{(v(O_a)+v(O_b))(v(O_b)-v(O_a)+v(Z_a))}{2v(Z_a)v(O_b)}$. When $v(Z_a)$ and $v(O_b)$ are fixed values, and $v(O_a)$ is a variable, this is a parabola with a maximum at $v(O_a)=\frac{v(Z_a)}{2}$. Plugging that in, we have the following upper bound on the approximation ratio: $\frac{(\frac{v(Z_a)}{2}+v(O_b))^2}{2v(Z_a)v(O_b)}=\frac{1}{2}+\frac{v(O_b)}{2v(Z_a)}+\frac{v(Z_a)}{8v(O_b)}$. Let us denote $x=\frac{v(O_b)}{v(Z_a)}$; then our upper bound is $\frac{1}{2}+\frac{x}{2}+\frac{1}{8x}$. We know that $x=\frac{v(O_b)}{v(Z_a)} \leq \frac{v(Z_b)}{v(Z_a)} < \frac{1}{\alpha-1}$ (since the mechanism did not end at option 1). We also know that $x=\frac{v(O_b)}{v(Z_a)}= \frac{2v(O_b)}{2v(Z_a)} \geq  \frac{v(O_b)+v(O_a)}{2v(Z_a)}=\frac{v(O)}{2v(Z_a)} \geq \frac{1}{2}$. So, to see how bad our upper bound can be, we maximize $\frac{1}{2}+\frac{x}{2}+\frac{1}{8x}$ over $x \in [\frac{1}{2},\frac{1}{\alpha-1}]$. Simple analysis shows that for $\alpha>1$, the maximum is  $\frac{1}{2}+\frac{1}{2(\alpha-1)}+\frac{\alpha-1}{8}$. We are therefore guaranteed approximation ratio $\alpha$ from our mechanism as long as $\frac{1}{2}+\frac{1}{2(\alpha-1)}+\frac{\alpha-1}{8} \leq \alpha$, which is easily seen to hold as long as $\alpha \geq \frac{5+4\sqrt{2}}{7}$.
\end{proof}

In the appendix, we show that $\frac{5+4\sqrt{2}}{7}$ was not an overestimate---this is tight for EQUAL-UTILITY. Finally, we note that EQUAL-UTILITY requires solving NP-hard problems: computing $OPT(X_1)$ and $OPT(X_2)$ means solving the knapsack problem, which is known to be NP-hard. In the appendix, we show that solving PROGRAM is NP-hard as well. Therefore, if running time is a consideration, EQUAL-UTILITY might be insufficient.\footnote{We do have approximation algorithms for these NP-hard problems: there is a known FPTAS for knapsack \cite{vazirani2013approximation}, and from the proof of Theorem \ref{thm:equtilstratproof} we can easily deduce an algorithm which produces $\frac{1}{1-\epsilon}\alpha$ approximation to PROGRAM  in time polynomial in $\frac{1}{\epsilon}$ whenever EQUAL-UTILITY needs to solve it. The problem is that if we use them, strategyproofness is technically violated in a strange fashion: for example, consider trying to compute $OPT(X_i)$. If $X_i$ contains few items, the designer solves the problem optimally, but if $X_i$ contains many items, then-- due to running time considerations-- she solves it approximately. Thus, if $X_i$ is large, agent $i$ might be better off only reporting the items included in the optimal solution, to avoid risking a suboptimal solution being chosen by the approximation algorithm. Therefore, despite the fact that the goals of the agent and the designer are fully aligned with each other in this case, the agent can potentially benefit from hiding all items except those in the optimal solution. Nevertheless, if she does so, it is in an attempt to \emph{assist}, rather than mislead, the designer by narrowing down the search space. Informally speaking, these strange violations could be eliminated if we allow the designer to ``listen" to the agent's solution proposals for such problems, and choose the best between her own computed solution and the agent's.}

\section{The PACIFY-THE-LIAR Mechanism}\label{sec:pacifytheliar}

We continue exploring the understating model. We now allow for a general number of agents $n$, however we restrict ourselves to an environment where there is only one bad apple---specifically, $n-1$ agents are assumed to be honest; we assume without loss of generality that agent $1$ is the manipulative agent.\footnote{Certainly, if the identity of the manipulative agent is public, assuming that it is agent $1$ is without loss of generality. Our results hold even if the honesty of an agent---whether or not that agent is manipulative---is private information of that agent, if we naturally extend our definition of mechanism to allow reporting of all private data, including honesty. Our lower bound will of course still hold. As for our upper bound, consider an augmented mechanism: if exactly one agent reports that they are manipulative, run PACIFY-THE-LIAR accordingly; otherwise, return $\varnothing$. As all the other agents answer that they are honest, it is easily seen that it is always weakly better for the manipulative agent to reveal herself as dishonest.} For this environment, we will provide a $\phi$-approximate deterministic strategyproof mechanism ($\phi=\frac{1+\sqrt{5}}{2} \approx 1.618$ is the golden ratio). In Section 6 we provide a matching lower bound.

Our deterministic mechanism, called PACIFY-THE-LIAR (Algorithm 3), begins similarly to EQUAL-UTILITY in the sense that it attempts to see if agent $1$ can guarantee an $\alpha$-approximation on her own, or if agents $2$ through $n$ can guarantee an $\alpha$-approximation together, without agent $1$. In the former case, we return $OPT(X_1)$, and in the latter case we return $OPT(\cup_{i \in N \backslash \{1\}} X_i)$. If neither of the cases hold, we look at a collection of solutions that guarantee $\alpha$-approximation, and attempt to ``pacify" agent $1$ by choosing her favorite solution within that collection.

\begin{algorithm}[t]
\SetAlgoNoLine
\KwIn{Sets of items $X_1,\ldots,X_n$, where $X_i \in \widehat{G_i}$; parameter $\alpha \geq 1$}
$Z_1 \leftarrow OPT(X_1)$, $Z_2 \leftarrow OPT(\cup_{i \in N \backslash \{1\}} X_i)$

\uIf{$v(Z_1) \geq \frac{1}{\alpha}(v(Z_1)+v(Z_2))$}{
\Return $Z_1$ (option 1)}
\uElseIf{$v(Z_2) \geq \frac{1}{\alpha}v(OPT(\cup_{i \in N}X_i))$}{\Return $Z_2$ (option 2)}
\Else{$S \leftarrow \{A \subseteq \cup_{i \in N} X_i:v(A) > \alpha v(Z_2)\}$

\Return $\argmax_{A \in S} v(A \cap X_1)$ (option 3)}
\caption{PACIFY-THE-LIAR}
\end{algorithm}

Note that if we reach option 3, $S$ is nonempty since we did not stop at option 2 (thus $S$ includes the optimal solution).
\begin{theorem}
PACIFY-THE-LIAR is strategyproof and $\alpha$-approximate for $\alpha \geq \phi$.
\end{theorem}

\begin{proof}
We begin by proving strategyproofness. If the algorithm ends at options 1 or 2, the proof is similar to case (1) in the proof of Theorem \ref{thm:equtilstratproof}. If the algorithm ends at option 3, the fact that agent $1$ cannot benefit from making the mechanism stop at an earlier option follows from a similar argument to the one in case (2) in the proof of Theorem \ref{thm:equtilstratproof}. Thus, all we need to show is that if the mechanism stops at option 3 under agent $1$'s misreport, agent $1$ does not benefit. Assume $X_1' \subset X_1$ and let $S'=\{A \subseteq X_1' \cup (\cup_{i \in N \backslash \{1\}} X_i):v(A) > \alpha v(Z_2)\}$. Then, note that $S' \subseteq S$, and for every $A \in S'$, $v(A \cap X_1')=v(A \cap X_1)$; therefore,  $\max_{A \in S'}v(A \cap X_1) \leq \max_{A \in S}v(A \cap X_1)$, and so agent $1$ does not benefit.

Now that we have established strategyproofness, let us analyze the approximation ratio. Clearly $\alpha$-approximation is guaranteed when the mechanism ends at options 1 or 2. So let us consider the case where the mechanism ends at option 3.  Let $A$ be the output. Since $A \in S$, $v(A)>\alpha v(Z_2)$. Since we did not stop at option 1, $v(Z_1) < \frac{1}{\alpha-1}v(Z_2)$, and hence $v(Z_1)+v(Z_2) < \frac{\alpha}{\alpha-1}v(Z_2)$. Therefore, $v(A) > \alpha v(Z_2) > (\alpha-1)(v(Z_1)+v(Z_2)) \geq (\alpha-1)v(OPT(\cup_{i \in N}X_i))$. Thus, we are guaranteed an $\alpha$-approximate mechanism if $\frac{1}{\alpha} \leq (\alpha-1)$, and this is easily seen to hold true for $\alpha \geq \phi$.
\end{proof}

Finally, we note that PACIFY-THE-LIAR, just like EQUAL-UTILITY, does not run in polynomial time.

\section{Lower Bounds}\label{sec:lowerbounds}

In this section we provide some lower bounds on the approximation ratios that strategyproof mechanisms can provide. Our lower bounds are robust---they all hold even when there are only two agents and only one is manipulative. In some cases, our lower bounds match what our mechanisms from the previous sections give us, proving their optimality. Throughout this section we assume that the ground sets are unrestricted.

We begin with the overstating model. 

\begin{theorem}
In the overstating model, there is no randomized strategyproof mechanism with a worst-case approximation ratio strictly smaller than $2$.
\end{theorem}

\begin{proof}
Let $f$ be a randomized strategyproof mechanism with worst-case approximation ratio $1<r<2$ (the case of $r=1$ has been covered in Example \ref{ex:abcd}). Consider ${\bf X}' \in \widehat{G_1} \times \cdots \times \widehat{G_n}$. In this instance,  $X_1'=\{a_1,\ldots,a_{M^2}\}$, where $v(a_j)=\frac{1}{M}$ and $s(a_j)=\frac{1}{M^2}$ for all $j=1,\ldots,M^2$; $X_2'=\{b\}$ where $v(b)=s(b)=1$; $X_i'=\varnothing$ for all $i > 2$, where $M$ is some very large integer. The optimal solution for this instance is $X_1'$, with optimal value $M$. If no item in $X_1'$ is chosen with probability strictly more than $p$, the approximation ratio is at least $\frac{M}{pM+(1-p)}$ (note that choosing any items from $X_1'$ excludes choosing $b$ and vice versa), and thus we must have $\frac{M}{pM+(1-p)} \leq r$, namely $p \geq \frac{M-r}{Mr-r}$. Thus it must be the case that some item $a_j$ is chosen with probability at least $q=\frac{M-r}{Mr-r}$.

Next, consider ${\bf X}$, which is identical to ${\bf X'}$ except $X_1=\{a_j\}$. In this instance, the optimal solution is $X_2$, with optimal value $1$. Due to strategyproofness, it must be the case that item $a_j$ is chosen with probability at least $q$ (since otherwise agent $1$ has an incentive to report $X_1'$ instead of $X_1$). Thus, the approximation ratio on this instance is at least $\frac{1}{\frac{q}{M}+(1-q)}=\frac{1}{\frac{M-r}{M^2r-Mr}+\frac{Mr-M}{Mr-r}}$. Sending $M \rightarrow \infty$, this becomes $\frac{r}{r-1}>r$ for $r \in (1,2)$. Contradiction.
\end{proof}

Thus, the theorem above shows that HALF-GREEDY is best possible in the overstating model. Next, we show that randomization is necessary for good approximation:

\begin{theorem}
In the overstating model, there is no deterministic strategyproof mechanism with a constant worst-case approximation ratio.
\end{theorem}

\begin{proof}
Let $f$ be a deterministic strategyproof mechanism with worst-case approximation ratio $r \in [1,\infty)$. Let ${\bf X'}$ be as in the proof of the previous theorem. Assume $M>r$. If the mechanism doesn't choose any item in $X_1'$, the approximation ratio is at least $\frac{M}{1}>r$. Thus, the mechanism chooses at least one item in $X_1'$, say $a_j$. Now, consider ${\bf X}$. On that profile, the mechanism must choose $a_j$ due to strategyproofness, leading to an approximation ratio of $\frac{1}{\frac{1}{M}}=M>r$. Contradiction.
\end{proof}

Next, we consider the understating model. We show that no deterministic strategyproof mechanism can give an approximation ratio better than $\phi$, thus proving the optimality of PACIFY-THE-LIAR in its specialized environment (recall that all of our lower bounds hold when all agents but one are honest). As this lower bound only uses $2$ agents, it also shows that EQUAL-UTILITY's (in its specialized environment of $n=2$) approximation ratio cannot be beaten by a deterministic mechanism.

\begin{theorem}\label{thm:dumlb}
In the understating model, no deterministic strategyproof mechanism can provide an approximation ratio better than $\phi$.
\end{theorem}

\begin{proof}
Let $f$ be a deterministic strategyproof mechanism with approximation ratio $r<\phi$. Consider the profile where $X_1=\{a\}$ with $s(a)=1$ and $v(a)=\phi$, $X_2=\{b\}$ with $s(b)=\frac{1}{2}$ and $v(a)=1$, and $X_i=\varnothing$ for all $i \geq 3$. To maintain approximation ratio $r$, we must have $a \in f({\bf X})$. Consider the profile ${\bf X}'$ that differs from ${\bf X}$ only in $X_1'=\{a,c\}$, where $s(c)=\frac{1}{2}$ and $v(c)=\phi-\epsilon$ for some small $\epsilon>0$. To maintain strategyproofness, we must have $a \in f({\bf X'})$. Thus, the approximation ratio on that profile is $\frac{\phi-\epsilon+1}{\phi}$, which can be made arbitrarily close to $\frac{\phi+1}{\phi}=\phi$. Therefore, the worst-case approximation ratio of $f$ cannot be better than $\phi$.
\end{proof}

Finally, we show a lower bound of $\frac{5\sqrt{5}-9}{2} \approx 1.09$ on the approximation ratio of randomized strategyproof mechanisms in the understating model. The importance of this result is that it shows that strategyproofness leads to a strict separation from optimality in the understating model, and thus it is reasonable to look at approximation mechanisms there.

\begin{theorem}\label{thm:rumlb}
In the understating model, no randomized strategyproof mechanism can provide a worst-case approximation ratio strictly better than $\frac{5\sqrt{5}-9}{2} \approx 1.09$.
\end{theorem}

\begin{proof}
Let $f$ be a randomized strategyproof mechanism which provides a worst-case approximation ratio $r<\frac{5\sqrt{5}-9}{2}$. Consider the profile where $X_1=\{a\}$ with $s(a)=1$ and $v(a)=\phi=\frac{1+\sqrt{5}}{2}$ (the golden ratio), $X_2=\{b\}$ with $s(b)=\frac{1}{2}$ and $v(b)=1$, and $X_i=\varnothing$ for all $i \geq 3$. Let $p=\mathcal{P}(a \in f({\bf X}))$. To maintain approximation ratio $r$, we must have (I) $p\phi+(1-p) \geq \frac{1}{r}\phi$. Now, consider profile ${\bf X}'=({\bf X}_{-1},X_1')$ where $X_1'=\{a,c\}$ and $v(c)=1$, $s(c)=\frac{1}{2}$. Let $p'=\mathcal{P}(a \in f({\bf X}'))$. To maintain approximation ratio $r$, we must have (II) $p'\phi+(1-p')2 \geq \frac{2}{r}$. To maintain strategyproofness, we must have (III) $p'\phi+(1-p') \geq p \phi$. Now, (I) gives $p \geq \frac{\frac{\phi}{r}-1}{\phi-1}$, (II) gives $p' \leq \frac{2-\frac{2}{r}}{2-\phi}$. (III) can be rewritten as $p'(\phi-1)+1-p\phi \geq 0$, and so this implies $\frac{2-\frac{2}{r}}{2-\phi}(\phi-1)+1-\phi\frac{\frac{\phi}{r}-1}{\phi-1} \geq 0$. Isolating $r$, this gives $r \geq \frac{5\sqrt{5}-9}{2}$, contradiction.
\end{proof}

We note that the fact that the golden ratio shows up in both the proofs of Theorems \ref{thm:dumlb} and \ref{thm:rumlb} was surprising to us, as in both cases we used general parameters and optimized for the maximal lower bound. We wonder if there is a simple explanation for this.

\section{Manipulating Sizes Instead of Existence}\label{sec:kqus}

In this section, we consider a model which we call `Known Quality Unknown Size' (KQUS). In this model, the true item profile ${\bf X}$ is known, and in fact $|X_i|=1$ for all $i \in N$; hence, we simply call agent $i$'s item $a_i$ and avoid using ${\bf X}$. Furthermore, for each $a_i$, $r_i=\frac{v(a_i)}{s(a_i)}$ is known, however $v(a_i)$ and $s(a_i)$ themselves are private information of agent $i$. When $r(a_i)$ is given, $s(a_i)$ determines $v(a_i)$, so we simply ask agent $i$ to report $s(a_i)$. Agent $i$ gets a utility of $v(a_i)$ if her item is chosen and $0$ if not; note that her utility from $a_i$ being chosen is $v(a_i)$ even if she misreports. In this model the ``quality" of an item is known, but its indivisible value and size are not.

Formally, a deterministic mechanism in this model is a function $f:\mathbb{R}_{+}^{2n} \rightarrow 2^{\{a_1,\ldots,a_n\}}$, which maps $({\bf r},{\bf s})$ to a subset of the items to be included in the knapsack. We will require feasibility ($s(f({\bf r},{\bf s})) \leq 1$) and strategyproofness ($v(f({\bf r},{\bf s}) \cap \{a_i\}) \geq v(f({\bf r},({\bf s}_{-i},s_i')) \cap \{a_i\})$ for all $i \in N$, $s_i' \in (0,1]$). We will also look for randomized strategyproof mechanisms. The adaptation is similar to before: $f$ is a random variable over $2^{\{a_1,\ldots,a_n\}}$, feasibility is $s(f({\bf r},{\bf s})) \leq 1$ surely, and strategyproofness is $\mathbb{E}[v(f({\bf r},{\bf s}) \cap \{a_i\})] \geq \mathbb{E}[v(f({\bf r},({\bf s}_{-i},s_i')) \cap \{a_i\})]$. For convenience of presentation, we will allow items with zero value (all of our proofs can easily be adjusted to get rid such items).

In this model, while HALF-GREEDY is not strategyproof (specifically, MAXIMUM-VALUE is not strategyproof), it can be easily modified to become strategyproof and remain $2$-approximate---in fact, the modified mechanism is also well known to be a $2$-approximation in non-strategic environments \cite{burke2005search}. $\succeq'$ is defined as before.

\begin{definition}
The NEXT mechanism is defined as follows: if $s(\cup_{i \in N} \{a_i\}) \leq 1$, return $\varnothing$, otherwise return $o_{\cup_{i \in N} \{a_i\}}$
\end{definition}

\begin{definition}
MODIFIED-HALF-GREEDY runs GREEDY with probability $\frac{1}{2}$ and NEXT with probability $\frac{1}{2}$. 
\end{definition}

MODIFIED-HALF-GREEDY still runs GREEDY with probability $\frac{1}{2}$, but otherwise it doesn't choose the item with the maximal value, but rather the first item to not make it into the knapsack in GREEDY.

\begin{theorem}
In KQUS, MODIFIED-HALF-GREEDY is strategyproof and $2$-approximate.
\end{theorem}

\begin{proof}
Fix agent $j \in N$. In MODIFIED-HALF-GREEDY, item $a_j$ is chosen with probability either $\frac{1}{2}$ or $0$. Specifically, an item is chosen with probability $\frac{1}{2}$ iff it is in $L_{\cup_{i \in N} \{a_i\}}(o_{\cup_{i \in N} \{a_i\}}) \cup \{o_{\cup_{i \in N} \{a_i\}}\}$ (if $s(\{a_1,\ldots,a_n\}) \leq 1$, then all items are chosen with probability $\frac{1}{2}$) . If item $a_j$ is chosen with probability $\frac{1}{2}$, agent $j$ has no incentive to manipulate. If item $a_j$ is chosen with probability $0$, then $o_{\cup_{i \in N} \{a_i\}} \succ' a_j$, so $s(a_j)$ has no impact on what $L_{\cup_{i \in N} \{a_i\}}(o_{\cup_{i \in N} \{a_i\}}) \cup \{o_{\cup_{i \in N} \{a_i\}}\}$ is, and thus agent $j$ cannot make the item get chosen. So strategyproofness is proven. $2$-approximation of MODIFIED-HALF-GREEDY is, as we mentioned, known.
\end{proof}

As in GREEDY, we made some careful choices here to preserve strategyproofness. Specifically, the choice for NEXT to return $\varnothing$ when $s(\cup_{i \in N} \{a_i\}) \leq 1$, despite us being able to include all items in the knapsack in that case, is crucial; if we indeed included all items in this case, strategyproofness would have been violated.

Next, a matching lower bound:

\begin{theorem}\label{thm:kquslowerboundrand}
No randomized strategyproof mechanism can provide a worst-case approximation ratio strictly better than $2$ in KQUS.
\end{theorem}

\begin{proof}
Let $f$ be a randomized strategyproof mechanism with approximation ratio $t<2$. Consider the case where $r_1=M$, $r_2=1$, and $r_i=0$ for $i \geq 3$, where $M$ is a large integer. Assume $s(a_1)=1$ and $s(a_2)=1$. Let $p_1$ be the probability with which item $a_1$ is chosen under $f$. Then $p_1M+(1-p_1) \geq \frac{1}{t}M$ to maintain approximation ratio $t$, and therefore $p_1 \geq \frac{\frac{M}{t}-1}{M-1}$. Now, consider the case where $s(a_1)=\frac{1}{M^2}$ (the rest of the data remains the same), and let $p_1'$ be the probability that item $a_1$ is chosen under $f$ in this case. To maintain strategyproofness, we must have $p_1=p_1'$. Therefore, to maintain approximation ratio $t$, we must have $p_1\frac{1}{M}+(1-p_1) \geq \frac{1}{t}$, which yields $p_1 \leq \frac{1-\frac{1}{t}}{1-\frac{1}{M}}$. Therefore, we must have $\frac{1-\frac{1}{t}}{1-\frac{1}{M}} \geq \frac{\frac{M}{t}-1}{M-1}$, namely $1+\frac{1}{M} \geq \frac{2}{t}$. However, as $M \rightarrow \infty$, the left hand side goes to $1$ and the right hand side remains $\frac{2}{t}>1$. Contradiction.
\end{proof}

Finally, we show that randomization is necessary for good approximation:

\begin{theorem}\label{thm:kquslowerbounddet}
No deterministic strategyproof mechanism can provide a constant worst-case approximation ratio in KQUS.
\end{theorem}

\begin{proof}
Let $f$ be a deterministic strategyproof mechanism with approximation ratio $t$. Consider the case where $r_1=M$, $r_2=1$, and $r_i=0$ for $i \geq 3$, where $M>t$. Assume $s(a_1)=1$ and $s(a_2)=1$. On this instance, to maintain approximation ratio $t$, $f$ must choose $\{a_1\}$. Now, consider the case where $s(a_1)=\frac{1}{M^2}$ (the rest of the data is the same). On this instance, to maintain approximation ratio $t$, $f$ must choose $\{a_2\}$. Thus, when agent $1$'s item is of size $\frac{1}{M^2}$, she will have an incentive to report its size to be $1$, violating strategyproofness.
\end{proof}

\section{Future Research}\label{sec:future}

There are several natural directions for the continuation of our research. First, all of our lower bounds hold even when there are only two agents and only one is manipulative. It will be interesting to know if having more manipulative agents necessarily increases the attainable worst-case approximation ratio under strategyproofness. Second, we did not provide a strategyproof mechanism for the full model, and whether one with a constant approximation ratio exists is an open problem. Third, we did not provide a strategyproof mechanism for the general understating model, only for special cases of it; a randomized strategyproof mechanism with a very large approximation ratio is given in \cite{DBLP:journals/corr/abs-1104-2872}, but it is unclear whether a smaller approximation ratio is attainable. Finally, the question of whether there exists a better randomized strategyproof mechanism than EQUAL-UTILITY for $2$ agents in the understating model is open, as our lower bound there is not tight.

\section*{Acknowledgements}
This work was supported in part by CUNY Collaborative Incentive Research Grant (CIRG 21) 2153 and PSC-CUNY Research Award 67665-00 45. The authors would like to thank Jay Sethuraman and Garud Iyengar for helpful discussions.

\bibliography{BIBLIOGRAPHY}
\bibliographystyle{plain}

\newpage
\appendix
\section{APPENDIX}

In this appendix, we provide results and proofs that were omitted from the paper.

\begin{theorem}
Let $r=\frac{5+4\sqrt{2}}{7}$. For every $\delta>0$, there exist an instance where EQUAL-UTILITY (with $\alpha=r$) provides an approximation ratio strictly larger than $r-\delta$.
\end{theorem}

\begin{proof}
Consider the profile where $X_1=\{a,b\}$ with $v(a)=s(a)=1$ and $v(b)=s(b)=\frac{1}{2}$, and $X_2=\{c\}$ where $v(c)=\frac{1}{r-1}-\epsilon$ where $\epsilon>0$ is small, and $s(c)=\frac{1}{2}$. EQUAL-UTILITY will reach option 2 on this instance. In this case, item $a$ is chosen with probability $p=\frac{\frac{1}{r-1}-\epsilon-\frac{1}{2}}{\frac{1}{r-1}-\epsilon+\frac{1}{2}}$. Items $b$ and $c$ are chosen with probability $1-p$. The mechanism's approximation ratio is $\frac{\frac{1}{r-1}-\epsilon+\frac{1}{2}}{p+(1-p)(\frac{1}{r-1}-\epsilon+\frac{1}{2})}$, which, as $\epsilon \rightarrow 0$, goes to $\frac{5+4\sqrt{2}}{7}$.
\end{proof}

\begin{theorem}
Solving PROGRAM is NP-hard.
\end{theorem}

\begin{proof}
We prove this by reduction from knapsack. Say we have an instance of the knapsack problem with set of items $I$. Assume without loss of generality that the knapsack's capacity in this instance is $\frac{1}{2}$ and that the sizes of the items are at most $\frac{1}{2}$ each. Let an optimal solution be $OPT^*$; we want to know whether or not $v(OPT^*) \geq k$ for some $k > 0$. Set $X_1=I$, and $X_2=\{a\}$ where $v(a)=k$ and $s(k)=\frac{1}{2}$, and solve PROGRAM on this instance. Note that agent $2$'s expected utility can never surpass $k$, as she only has one item and its value is $k$. Thus, the optimal value of $PROGRAM$ is at most $2k$, since the utilities must be equal. We claim that the optimal value of $PROGRAM$ is exactly $2k$ iff $v(OPT^*) \geq k$.
\begin{enumerate}
\item If $v(OPT^*) \geq k$, then there is a unique solution $A$ to PROGRAM where $A \in \{OPT^* \cup \{a\},\{a\}\}$ surely (as $v(X_1 \cap (OPT^* \cup \{a\}))=v(OPT^*)\geq k=v(X_2 \cap (OPT^* \cup \{a\}))$). Since $a$ is chosen with probability $1$, agent 2 gets an expected utility of exactly $k$ here, and hence so does agent 1. Thus, $\mathbb{E}[v(A)]=2k$, and therefore the optimal value of PROGRAM is at least $2k$, and thus is exactly $2k$.
\item If $v(OPT^*) < k$, note that whenever $S \subseteq X_1 \cup X_2$ where $v(S) \leq 1$, if $a \in S$, then $v(S \cap X_1) \leq  v(OPT^*) < k$ (since if $a \in S$, there is only capacity $\frac{1}{2}$ left for agent $1$'s items). Thus, it follows that for every solution $A$ of PROGRAM, there is a nonzero probability that $a \notin A$ (otherwise agent $1$'s expected utility must be strictly less than $k$, and agent $2$'s expected utility is exactly $k$). Thus agent $2$'s expected utility is strictly less than $k$, and therefore the optimal value of PROGRAM is strictly less than $2k$.
\end{enumerate}
\end{proof}

\end{document}